\documentclass[10pt]{article}
\usepackage{amsmath,amsfonts}
\usepackage{amssymb,amscd,amsthm}
\usepackage{diagbox}
\usepackage{makecell}
\usepackage{authblk}
\usepackage[all]{xy}
\usepackage{placeins}
\usepackage{geometry}
\usepackage{tabularx}
\usepackage{graphicx}
\usepackage{verbatim}
\usepackage[perpage,symbol*]{footmisc}
\usepackage[justification=centering]{caption}
\usepackage{longtable}
\usepackage{multirow}
\usepackage{indentfirst}  
\usepackage{enumitem}
\usepackage{hyperref}
\usepackage{bm}
\usepackage[figuresright]{rotating}
\usepackage{makecell}
\usepackage{subcaption}
\usepackage{rotating}
\theoremstyle{definition}
\newtheorem{definition}{Definition}

\newtheorem{lemma}{\textbf{Lemma}}
\newtheorem{theorem}{\textbf{Theorem}}

\newtheorem{corollary}{\textbf{Corollary}}
\usepackage{algorithm}
\usepackage{algpseudocode}
\newtheorem{example}{\textbf{Example}}

\begin{document}

	\baselineskip 17pt
\title{\Large\bfseries Bounds and Constructions of Maximum Toroidal Distance Codes}

\author{
	Pengjie Zhong\textsuperscript{1},
	Jinquan Luo\textsuperscript{2},
	and Yufeng Song\textsuperscript{3}
}

\date{}

	\maketitle
	
	\begingroup
	\renewcommand{\thefootnote}{}
	\footnotetext{
		
		Pengjie Zhong\textsuperscript{1} is with the School of Mathematics and Statistics, Central China Normal University, Wuhan, China.
		
		Jinquan Luo\textsuperscript{2} is with the School of Mathematics and Statistics \& Hubei Key Laboratory of Mathematical Sciences, Central China Normal University, Wuhan 430079, China.
		
		Yufeng Song\textsuperscript{3} is with the Department of Mathematics, Nanjing University of Aeronautics and Astronautics, Nanjing 210000, China.

		\textit{Email:} luojinquan@mail.ccnu.edu.cn(Jinquan Luo).
		
		\textit{Email:} pengjiezhong.math@qq.com(Pengjie Zhong).
		
		\textit{Email:} yufengsong.math@qq.com(Yufeng Song).
	}
	\endgroup
	\date{}
	\maketitle
	
	\qquad
	
	\begin{abstract}
	In lattice-based cryptographic schemes, both encoded messages and accumulated decryption noise are represented in a modulo $q$ space. Therefore, it is natural to study toroidal distances and maximum toroidal distance (MTD) codes.
	In this paper, we derive some upper bounds for minimum toroidal distance of a code, including a Plotkin-type bound, a local ball--Plotkin bound, and a Delsarte linear programming bound. We also exhibit examples showing that these bounds are sharp in some cases.
	Moreover, we present several code constructions with good minimum distance, some of which are MTD codes. For \(\ell=2\), we obtain a family of four-point MTD codes in \(\mathbb Z_q^2\). For $\ell=4$, we propose a general code construction and exhibit several explicit instances for specific values of $q$, some of which are proven to be MTD codes. For \(\ell=8\),  using the \(E_8\) lattice, we construct codes
	$C=2mE_8\cap \mathbb Z_q^8$, where $q=4m$
	and show that they are MTD codes. These results give explicit optimal constructions of MTD codes for \(\ell=2,4,8\). In the case $\ell=16$, we construct a code with minimum toroidal distance $3$ for $q=4$, while the known upper bound in this case is $2\sqrt{3}$.  Our main tools are geometric and linear programming methods.
	\end{abstract}
	
	Keywords{: toroidal distance, maximum toroidal distance code, lattice code, Plotkin bound, Delsarte linear programming bound. }

\section{Introduction}

In modular decoding problems for lattice-based schemes, the codebook is naturally viewed as a subset of the discrete torus $\mathbb Z_q^\ell$. Recent work on maximum toroidal distance (MTD) codes makes this explicit by formulating encoding as the selection of $2^\ell$ torus points and by maximizing the minimum $L_2$-norm toroidal distance~\cite{LiuSakzad2026}. Closely related lattice-quantizer approaches further reveal that Module Learning with Errors (M-LWE) based key reconciliation and Kyber encoding are closely tied to modulo $q$ codebook geometry and lattice packing~\cite{LiuSakzad2024,DesignsCodes2025}. In this sense, designing a good codebook is naturally interpreted as a finite packing problem on $\mathbb Z_q^\ell$.

On the cryptographic side, the Module-Lattice-Based Key-Encapsulation Mechanism (ML-KEM) standard is based on the M-LWE problem~\cite{FIPS203}, and the Kyber specification shows that noise growth and decryption failures play a central role in parameter selection~\cite{KyberSpec}. The lattice-quantizer framework makes this connection more concrete: it derives explicit decryption failure rate (DFR) bounds for generic M-LWE based key reconciliation mechanisms and shows that improved lattice codebooks can reduce the DFR without changing the underlying security parameters~\cite{LiuSakzad2024}. This motivates the use of the minimum toroidal distance as a natural geometric criterion for evaluating the resilience of modular codebooks against noise perturbations.

Extremal problems of this type are traditionally studied using packing arguments and linear programming methods. Delsarte's association-scheme framework provides a general linear programming (LP) approach for bounding code parameters~\cite{Delsarte73,DelsarteLevenshtein1998}, while subsequent work interprets LP bounds from packing and covering perspectives~\cite{NavonSamorodnitsky}. For the Lee metric, several refinements further strengthen LP bounds and improve the computational treatment of the linear constraints arising from Lee compositions~\cite{AstolaTabus,GabrysKiahMilenkovic2015,Sole88}. Since the toroidal distance can be regarded as a natural extension of the Lee distance, optimization problems for the toroidal distance can also be viewed as extremal geometric problems with both packing and LP formulations.

However, compared with the well-developed theories of codes in the Hamming and Lee metrics, the theory of toroidal distance codes remains relatively underdeveloped. Only a limited number of explicit MTD constructions are known, and exact MTD values are still unknown for many finite parameter regimes. This gap indicates that the theory of toroidal distance codes is far from complete and motivates our investigation of sharper upper bounds and new constructions.

The rest of the paper is organized as follows. In Section~\ref{pre}, we introduce some notations and preliminary materials used throughout the paper. In Section~\ref{bounds}, we study the connections among the Hamming, Lee, and toroidal distances. We also derive several upper bounds on the minimum toroidal distance, including a Plotkin-type bound, a local ball--Plotkin bound, and a Delsarte LP bound. In addition, we provide examples showing that these bounds are tight in certain cases. In Section~\ref{MTDcode}, we present several constructions of MTD codes for $\ell=2,4,8$ and prove that these codes are MTD. For $\ell=16$ and $q=4$, we construct a code with minimum toroidal distance $3$, while the known upper bound in this case is $2\sqrt{3}$.
\section{Preliminaries}	
\label{pre}
	
	Let $\bm{x}=(x_1,\dots,x_n)\in \mathbb R^n$. The $\ell_1$-norm of $\bm{x}$ is 
	\[
	\|\bm{x}\|_1=\sum_{i=1}^n |x_i|,
	\]
	and the $\ell_2$-norm of $\bm{x}$ is 
	\[
	\|\bm{x}\|_2=\left(\sum_{i=1}^n x_i^2\right)^{1/2}.
	\]

Throughout this paper, the residue class ring \(\mathbb Z_q\) is represented by \(\{0,1,\ldots,q-1\}\).
All additions and
subtractions in $\mathbb Z_q^n$ are taken modulo $q$. 
\begin{definition}[Three distances]
	For $a,b\in\mathbb Z_q$, define the Lee distance between $a$ and $b$ by
\[
d_L(a,b):=\min\{|a-b|,\ q-|a-b|\}.
\]

For $\bm{x}=(x_1,x_2,\dots,x_n)$, $\bm{y}=(y_1,y_2,\dots,y_n) \in \mathbb Z_q^n$, the \emph{Hamming}, \emph{Lee} and \emph{toroidal distances} between $\bm{x}$ and $\bm{y}$ are defined by

\begin{enumerate}[label=(\arabic*)]
	\item $d_H(\bm{x},\bm{y})
	:=
	\bigl|\{\,i\in\{1,2,\dots,n\}:x_i\neq y_i\,\}\bigr|;$
	\item $d_L(\bm{x},\bm{y})
	:=
	\sum_{i=1}^n d_L(x_i,y_i);$
	\item $d_T(\bm{x},\bm{y})
	:=
	\left(\sum_{i=1}^n d_L(x_i,y_i)^2\right)^{\frac{1}{2}}.$
	
\end{enumerate}		

For a code $C\subseteq \mathbb Z_q^n$, \emph{the minimum Hamming, Lee, and toroidal distances} of $C$ are defined by
\begin{enumerate}[label=(\arabic*)]
	\item $	d_H(C):=\min_{\substack{\bm{x},\bm{y}\in C\\ \bm{x}\neq \bm{y}}}
	d_H(\bm{x},\bm{y});$
	\item $	d_L(C):=\min_{\substack{\bm{x},\bm{y}\in C\\ \bm{x}\neq \bm{y}}}
	d_L(\bm{x},\bm{y});$
	\item $	d_T(C):=\min_{\substack{\bm{x},\bm{y}\in C\\ \bm{x}\neq \bm{y}}}
	d_T(\bm{x},\bm{y}).$
\end{enumerate}

	In particular, for a codeword $\bm u$, we write
	\[
	|\bm u|_T := d_T(\bm u,\bm 0).
	\]
	\end{definition}

	Let \(C\subseteq \mathbb Z_q^n\) be a code with cardinality \(|C|=M\) and
	minimum toroidal distance \(d_T\). Then \(C\) is called an
	\((n,M,d_T)_q\) \emph{toroidal distance code}.
\section{Bounds on Minimum Toroidal Distance}
	\label{bounds}
	
	\subsection{Comparison of Three Distances}

	\begin{theorem}
		\label{connections}
		For $C \subseteq \mathbb{Z}_q^n$ , then
	
\begin{enumerate}[label=(\roman*)]
	\item $d_T(C) \leq \lfloor \frac{q}{2} \rfloor \sqrt{ d_H(C)},$ 
	the equality holds iff there exist distinct codewords
	\(\bm{x},\bm{y}\in C\) such that 	
\[
\left\{
\renewcommand{\arraystretch}{1.35}
\begin{array}{@{\quad}l}
	d_H(\bm{x},\bm{y})=d_H(C),\\
	d_T(\bm{x},\bm{y})=d_T(C),\\
	d_L(x_i,y_i)\in\{0,\lfloor \frac{q}{2} \rfloor\}  \text{ for all } 1 \leq i \leq n.
\end{array}
\right.
\]
	\item\label{item:second-condition} $\frac{d_L(C)}{\sqrt{n}} \leq d_T(C) \leq d_L(C),$
	the left equality holds iff there exist distinct codewords
	\(\bm{x},\bm{y}\in C\) such that
	\[
	\left\{
	\renewcommand{\arraystretch}{1.35}
	\begin{array}{@{\quad}l}
		d_L(\bm{x},\bm{y})=d_L(C),\\
		d_T(\bm{x},\bm{y})=d_T(C),\\
		d_L(x_i,y_i)=d_L(x_j,y_j) \text{ for all } 1 \leq i,j \leq n.
	\end{array}
	\right.
	\]
	The right equality holds iff there exist distinct codewords
	\(\bm{x},\bm{y}\in C\) such that
	\[
	\left\{
	\renewcommand{\arraystretch}{1.35}
	\begin{array}{@{\quad}l}
		d_L(\bm{x},\bm{y})=d_L(C),\\
		d_T(\bm{x},\bm{y})=d_T(C),\\
		d_L(x_i,y_i)=0 \text{ except for only one } i .
	\end{array}
	\right.
	\]
	
	\item\label{item:third-condition} 
	$d_T(C)\le \sqrt{\left\lfloor \frac{q}{2}\right\rfloor d_L(C)}$,
	the equality holds iff there exist distinct codewords
	\(\bm{x},\bm{y}\in C\) such that
	\[
	\left\{
	\renewcommand{\arraystretch}{1.35}
	\begin{array}{@{\quad}l}
		d_L(\bm{x},\bm{y})=d_L(C),\\
		d_T(\bm{x},\bm{y})=d_T(C),\\
		d_L(x_i,y_i)\in\{0,\lfloor \frac{q}{2} \rfloor\} \text{ for all } 1 \leq i \leq n.
	\end{array}
	\right.
	\]

\end{enumerate}

	\end{theorem}
	\begin{proof}
		
\begin{enumerate}[label=(\roman*)]
\item For any distinct \(\bm{x},\bm{y}\in C\), put
\[
a_i=d_L(x_i,y_i),\qquad 1\le i\le n.
\]
Then \(0\le a_i\le \lfloor q/2\rfloor\), and \(a_i=0\) exactly when
\(x_i=y_i\). Hence
\begin{equation}
	\label{distance11}
d_T(\bm{x},\bm{y})^2
=
\sum_{i=1}^n a_i^2
\le
\left\lfloor \frac{q}{2}\right\rfloor^2
\#\{i:a_i\neq 0\}.
\end{equation}
Since \(\#\{i:a_i\neq 0\}=d_H(\bm{x},\bm{y})\), we get
\[
d_T(\bm{x},\bm{y})
\le
\left\lfloor \frac{q}{2}\right\rfloor
\sqrt{d_H(\bm{x},\bm{y})}.
\]

Now choose distinct codewords \(\bm{x},\bm{y}\in C\) such that
\[
d_H(\bm{x},\bm{y})=d_H(C).
\]
Then
\begin{equation}
	\label{distance12}
	d_T(C)
	\le
	d_T(\bm{x},\bm{y})
	\le
	\left\lfloor \frac{q}{2}\right\rfloor
	\sqrt{d_H(C)}.
\end{equation}

By \eqref{distance11} and \eqref{distance12}, we obtain that equality holds iff there exist distinct codewords
\(\bm{x},\bm{y}\in C\) such that
\[
\left\{
\renewcommand{\arraystretch}{1.35}
\begin{array}{@{\quad}l}
	d_H(\bm{x},\bm{y})=d_H(C),\\
	d_T(\bm{x},\bm{y})=d_T(C),\\
	d_L(x_i,y_i)\in\left\{0,\left\lfloor q/2\right\rfloor\right\}
	 \text{ for all } 1 \leq i \leq n.
\end{array}
\right.
\]

	\item For any \(\bm{x},\bm{y}\in \mathbb Z_q^n\), put
	\[
	a_i=d_L(x_i,y_i),\qquad 1\le i\le n.
	\]

	Since
	\begin{equation}
		\label{threedistance}
		\begin{aligned}
			d_L(\bm{x},\bm{y})^2
			=
			\left(\sum_{i=1}^n a_i\right)^2  
			&=
			\sum_{i=1}^n a_i^2
			+
			2\sum_{1\le i<j\le n}a_i a_j  \\
			&\ge
			\sum_{i=1}^n a_i^2
			=
			d_T(\bm{x},\bm{y})^2,
		\end{aligned}
	\end{equation}
	we have
	\[
	d_T(\bm{x},\bm{y})\le d_L(\bm{x},\bm{y}).
	\]
	By \eqref{threedistance}, the equality holds iff
	\[
	a_i a_j=0\qquad \text{for all }i<j,
	\]
which is equivalent to saying that at most one of the \(a_i\) is nonzero.
	
	On the other hand, 
	\[
	\frac{1}{n}\sum_{i=1}^n a_i
	\le
	\left(\frac{1}{n}\sum_{i=1}^n a_i^2\right)^{1/2}.
	\]
	Hence
$$
		d_L(\bm{x},\bm{y})
		=
		\sum_{i=1}^n a_i
		\le
		\sqrt n
		\left(\sum_{i=1}^n a_i^2\right)^{1/2}
		=
		\sqrt n\,d_T(\bm{x},\bm{y}).
	$$
	The equality holds iff
	\[
	a_1=a_2=\cdots=a_n.
	\]
	Therefore,
	\[
	\frac{d_L(\bm{x},\bm{y})}{\sqrt n}
	\le
	d_T(\bm{x},\bm{y})
	\le
	d_L(\bm{x},\bm{y}).
	\]
	Taking minimum over all distinct codewords \(\bm{x},\bm{y}\in C\) gives
	\[
	\frac{d_L(C)}{\sqrt n}
	\le
	d_T(C)
	\le
	d_L(C).
	\]
	
	The left equality holds iff there exist distinct codewords
	\(\bm{x},\bm{y}\in C\) such that
	\[
	\left\{
	\renewcommand{\arraystretch}{1.35}
	\begin{array}{@{\quad}l}
		d_L(\bm{x},\bm{y})=d_L(C),\\
		d_T(\bm{x},\bm{y})=d_T(C),\\
		d_L(x_i,y_i)=d_L(x_j,y_j) \text{ for all } 1 \leq i,j \leq n.
	\end{array}
	\right.
	\]
	Indeed, if the left equality holds, choose
	\(\bm{x},\bm{y}\in C\) with
	\(d_T(\bm{x},\bm{y})=d_T(C)\). Then
	\[
	d_L(C)
	\le
	d_L(\bm{x},\bm{y})
	\le
	\sqrt n\,d_T(\bm{x},\bm{y})
	=
	\sqrt n\,d_T(C)
	=
	d_L(C).
	\]
	Thus both inequalities above are equalities, which implies
	\(d_L(\bm{x},\bm{y})=d_L(C)\) and
	\(d_L(x_i,y_i)=d_L(x_j,y_j)\) for all \(i,j\).
	
	The right equality of \ref{item:second-condition} holds iff there exist distinct codewords
	\(\bm{x},\bm{y}\in C\) such that
	\[
	\left\{
	\renewcommand{\arraystretch}{1.35}
	\begin{array}{@{\quad}l}
		d_L(\bm{x},\bm{y})=d_L(C),\\
		d_T(\bm{x},\bm{y})=d_T(C),\\
		d_L(x_i,y_i)=0 \text{ except for only one } i.
	\end{array}
	\right.
	\]
	Indeed, if the right equality of \ref{item:second-condition} holds, choose
	\(\bm{x},\bm{y}\in C\) with
	\(d_L(\bm{x},\bm{y})=d_L(C)\). Then
	\[
	d_T(C)
	\le
	d_T(\bm{x},\bm{y})
	\le
	d_L(\bm{x},\bm{y})
	=
	d_L(C)
	=
	d_T(C).
	\]
	Thus the above are all equalities, which implies
	\(d_T(\bm{x},\bm{y})=d_T(C)\) and at most one of 
	\(d_L(x_i,y_i)\) is nonzero.

		\item For any \(\bm{x},\bm{y}\in\mathbb Z_q^n\),
		\[
		d_T(\bm{x},\bm{y})^2
		=
		\sum_{i=1}^n d_L(x_i,y_i)^2
		\le
		\left\lfloor \frac{q}{2}\right\rfloor
		\sum_{i=1}^n d_L(x_i,y_i)
		=
		\left\lfloor \frac{q}{2}\right\rfloor d_L(\bm{x},\bm{y}).
		\]
		Now choose distinct codewords \(\bm{x},\bm{y}\in C\) such that
		\(d_L(C)=d_L(\bm{x},\bm{y})\). Then
		\[
		d_T(C)^2
		\le d_T(\bm{x},\bm{y})^2
		\le
		\left\lfloor \frac{q}{2}\right\rfloor d_L(\bm{x},\bm{y})
		=
		\left\lfloor \frac{q}{2}\right\rfloor d_L(C).
		\]
		Therefore,
		\[
		d_T(C)
		\le
		\sqrt{\left\lfloor \frac{q}{2}\right\rfloor d_L(C)}.
		\]
		
		Equality holds iff there exist distinct codewords
		\(\bm{x},\bm{y}\in C\) such that
		\[
		\left\{
		\renewcommand{\arraystretch}{1.35}
		\begin{array}{@{\quad}l}
			d_L(\bm{x},\bm{y})=d_L(C),\\
			d_T(\bm{x},\bm{y})=d_T(C),\\
			d_L(x_i,y_i)\in\left\{0,\left\lfloor q/2\right\rfloor\right\}
			 \text{ for all } 1 \leq i \leq n.
		\end{array}
		\right.
		\]
		
		Indeed, if equality holds, choose distinct codewords
		\(\bm{x},\bm{y}\in C\) such that
		\(d_L(\bm{x},\bm{y})=d_L(C)\). Then
		\[
		d_T(C)^2
		\le
		d_T(\bm{x},\bm{y})^2
		\le
		\left\lfloor \frac{q}{2}\right\rfloor d_L(\bm{x},\bm{y})
		=
		\left\lfloor \frac{q}{2}\right\rfloor d_L(C)
		=
		d_T(C)^2.
		\]
		Thus the above are all equalities. Hence
		\(d_T(\bm{x},\bm{y})=d_T(C)\). Moreover, the second inequality is an
		equality iff
		\[
		d_L(x_i,y_i)^2
		=
		\left\lfloor \frac{q}{2}\right\rfloor d_L(x_i,y_i)
		\]
		 for every \(i\), which implies
		\[
		d_L(x_i,y_i)\in
		\left\{0,\left\lfloor q/2\right\rfloor\right\}
		\quad \text{for all } i.
		\]
		The converse also holds by the same reasoning.
	\end{enumerate}	
	\end{proof}

\subsection{Global Plotkin-Type Bound}

	The following result was given in \cite{WynerGraham1968}.
	
	\begin{lemma}
		\label{leelema}
		Let $C\subseteq\mathbb Z_q^n$ be a code with
	$|C|=M$. Then
		\[
		d_L(C)\le 	\begin{cases}
			\dfrac{qMn}{4(M-1)},  &q \text{ even},\\[6mm]
			\dfrac{(q^2-1)Mn}{4q(M-1)}, &q \text{ odd}.
		\end{cases}
		\]

	\end{lemma}

\begin{theorem}[Global Plotkin-Type Bound]
	\label{Bq}
	Let $C\subseteq\mathbb Z_q^n$ be a toroidal-distance code with
	$|C|=M\ge 2$. Then
	\[
	d_T(C)
	\le
	\begin{cases}
		\dfrac{q}{2\sqrt2}\sqrt{\dfrac{Mn}{M-1}}, & q \text{ even},\\[6mm]
		\sqrt{\dfrac{(q-1)(q^2-1)}{8q}\cdot\dfrac{Mn}{M-1}}, & q \text{ odd}.
	\end{cases}
	\]
	In the  case $q$ even, equality holds for
	\[
	C=\left\{\bm a_0+\frac q2\bm b \bmod q:\bm b\in B\right\},
	\]
	where $\bm a_0\in\mathbb Z_q^n$ is fixed and
	$B\subseteq\mathbb Z_2^n$ satisfies $|B|=M$ and
	$
	d_H(B)=\frac{Mn}{2(M-1)}
	$.
\end{theorem}

\begin{proof}
By the third assertion of Theorem~\ref{connections}, we have
\[
d_T(C)\le
\sqrt{\left\lfloor \frac q2\right\rfloor d_L(C)}.
\]
Combining this with Lemma~\ref{leelema}, we obtain
\[
d_T(C)
\le
\begin{cases}
	\dfrac{q}{2\sqrt2}\sqrt{\dfrac{Mn}{M-1}}, & q \text{ even},\\[6mm]
	\sqrt{\dfrac{(q-1)(q^2-1)}{8q}\cdot\dfrac{Mn}{M-1}}, & q \text{ odd}.
\end{cases}
\]
	
	It remains to explain the equality condition in the case $q$ even.
	Suppose
	\[
	C=\left\{\bm a_0+\frac q2\bm b \bmod q:\bm b\in B\right\},
	\]
	where $B\subseteq\mathbb Z_2^n$ and
	$d_H(B)=Mn/(2(M-1))$. For any distinct
	$\bm b,\bm b'\in B$, 
	\[
	d_T\left(\bm a_0+\frac q2\bm b,\bm a_0+\frac q2\bm b'\right)^2
	=d_T\left(\frac q2\bm b, \frac q2\bm b' \right)^2
	=
	\left(\frac q2\right)^2 d_H(\bm b,\bm b').
	\]
	Thus
	\[
	d_T(C)^2
	=
	\left(\frac q2\right)^2 d_H(B)
	=
	\frac{q^2}{4}\cdot \frac{Mn}{2(M-1)}
	=
	\frac{q^2}{8}\cdot\frac{Mn}{M-1}.
	\]
\end{proof}

	\begin{example}
		Let \(q\) be even, and let \(S_r\) be the binary simplex code with parameters
		\[
		[\,2^r-1,\ r,\ 2^{r-1}\,].
		\]
		Every nonzero codeword of \(S_r\) have Hamming weight \(2^{r-1}\) and any two distinct codewords of \(S_r\) has Hamming distance \(2^{r-1}\).
		Replacing \(1\) by \(\frac{q}{2}\) in every codeword, and repeating each coordinate \(t\) times, we obtain a code
		\[
		C\subseteq \mathbb Z_q^{\,t(2^r-1)}
		\]
		with \(|C|=2^r\).
		
		Since any two distinct codewords of \(C\) differ in exactly
		\(t2^{r-1}\) coordinates, and in each such coordinate the Lee distance is
		\(\frac q2\), we have
		\[
		d_T(C)=\sqrt{t\,2^{r-1}\left(\frac{q}{2}\right)^2}=\,2^{\frac{r-3}{2}}q\sqrt{t}.
		\]
		On the other hand, by Theorem \ref{Bq}, we obtain
		$$d_T(C)\leq\dfrac{q}{2\sqrt2}\sqrt{\dfrac{2^rt(2^r-1)}{2^r-1}}=\,2^{\frac{r-3}{2}}q\sqrt{t}.$$
				Therefore, the code \(C\) attains the upper bound in Theorem~\ref{Bq}.
	\end{example}

\subsection{Local Ball--Plotkin Bound}
For $\bm z_0\in\mathbb Z_q^n$ and $\rho\in\mathbb Z_{\ge 0}$, define the ball with center $\bm z_0$ and toroidal radius $\sqrt{\rho}$ to be
\[
B(\bm z_0,\rho)
=
\{\bm x\in\mathbb Z_q^n:d_T(\bm x,\bm z_0)^2\le \rho\}.
\]
By translation invariance, the cardinality of $B(\bm z_0,\rho)$ is independent
of the center $\bm z_0$. We denote this common cardinality by $V_q(n,\rho)$.

Let
\[
m_r=
\#\{a\in\mathbb Z_q:d_L(a,0)=r\}.
\]
More explicitly,
	\[
m_r=
\begin{cases}
	1, & r=0 \text{ or } \left(r=\dfrac q2 \text{ and } q \text{ is even}\right),\\
	2, & 1\le r<\dfrac q2.
\end{cases}
\]
Then
\[
V_q(n,\rho)
=
\sum_{\substack{0\le r_1,\ldots,r_n\le \lfloor q/2\rfloor\\
		r_1^2+\cdots+r_n^2\le \rho}}
m_{r_1}\cdots m_{r_n}.
\]
Equivalently, this number can be computed by the generating function
\[
V_q(n,\rho)
=
\sum_{j=0}^{\rho}
[x^j]\left(\sum_{r=0}^{\lfloor q/2\rfloor}m_r x^{r^2}\right)^n ,
\]
where $[x^j]f(x)$ denotes the coefficient of the monomial
$x^j$ in $f(x)$.

\begin{lemma}
	\label{lem:lambda-toroidal-comparison}
	Let \(S\subseteq\mathbb Z_q^n\) be a finite nonempty set with
	\(|S|=N\), and fix
	\(\bm z_0=(z_1,\dots,z_n)\in\mathbb Z_q^n\). Then
	\[
	\sum_{\bm x,\bm y\in S}d_T(\bm x,\bm y)^2
	\le
	2N\sum_{\bm x\in S}d_T(\bm x,\bm z_0)^2.
	\]
\end{lemma}

\begin{proof}
	For each \(\bm x=(x_1,\dots,x_n)\in S\) and each coordinate
	\(1\le s\le n\), choose an integer representative \(\bar{x}_s\) of
	\(x_s-z_s\pmod q\) such that
	\[
	|\bar{x}_s|=d_L(x_s,z_s).
	\]
	Then, for any \(\bm x,\bm y\in S\),
	\[
	d_L(x_s,y_s)
	=
	d_L(x_s-z_s,y_s-z_s)
	\le
	|\bar{x}_s-\bar{y}_s|.
	\]
	Hence
\begin{equation}
\label{xs1}
	\begin{aligned}
		\sum_{\bm x,\bm y\in S}d_T(\bm x,\bm y)^2
		&=
		\sum_{\bm x,\bm y\in S}\sum_{s=1}^n d_L(x_s,y_s)^2  \\
		&\le
		\sum_{s=1}^n
		\sum_{\bm x,\bm y\in S}
		(\bar{x}_s-\bar{y}_s)^2 .
	\end{aligned}
\end{equation}
	For each fixed \(s\), we have
	\begin{equation}
\label{xs2}
	\begin{aligned}
		\sum_{\bm x,\bm y\in S}
		(\bar{x}_s-\bar{y}_s)^2
		&=
		\sum_{\bm x\in S}\sum_{\bm y\in S}\bar{x}_s^2
		-
		2\sum_{\bm x\in S}\sum_{\bm y\in S}\bar{x}_s\bar{y}_s
		+
		\sum_{\bm x\in S}\sum_{\bm y\in S}\bar{y}_s^2 \\
		&=
		N\sum_{\bm x\in S}\bar{x}_s^2
		-
		2\left(\sum_{\bm x\in S}\bar{x}_s\right)
		\left(\sum_{\bm y\in S}\bar{y}_s\right)
		+
		N\sum_{\bm y\in S}\bar{y}_s^2 \\
		&=
		2N\sum_{\bm x\in S}\bar{x}_s^2
		-
		2\left(\sum_{\bm x\in S}\bar{x}_s\right)^2 \\
		&\le
		2N\sum_{\bm x\in S}\bar{x}_s^2 .
	\end{aligned}
	\end{equation}
	Combining \eqref{xs1} and \eqref{xs2}, we obtain
	\[
	\begin{aligned}
		\sum_{\bm x,\bm y\in S}d_T(\bm x,\bm y)^2
		&\le
		2N\sum_{s=1}^n\sum_{\bm x\in S}\bar{x}_s^2 \\
		&=
		2N\sum_{\bm x\in S}\sum_{s=1}^n d_L(x_s,z_s)^2 \\
		&=
		2N\sum_{\bm x\in S}d_T(\bm x,\bm z_0)^2.
	\end{aligned}
	\]
\end{proof}

We now apply the preceding comparison to subsets of codewords contained
in a toroidal ball.
	\begin{theorem}[Local Ball--Plotkin Bound]
		\label{thm:local-ball-plotkin}
		Let $C\subseteq\mathbb Z_q^n$ be a code of size $M\ge2$. For each
		$\rho\in\mathbb Z_{\ge0}$, set
		\[
		L_\rho=\left\lceil\frac{M V_q(n,\rho)}{q^n} \right\rceil.
		\]
		Then
		\[
		d_T(C)^2
		\le
		\min_{\substack{\rho\in\mathbb Z_{\ge0}\\ L_\rho\ge2}}
		2\rho\frac{L_\rho}{L_\rho-1}.
		\]
	\end{theorem}
	
	\begin{proof}
		Fix $\rho\in\mathbb Z_{\ge0}$ with $L_\rho\ge2$. For
		$\bm z\in\mathbb Z_q^n$, recall that
		\[
		B(\bm z,\rho)
		=
		\{\bm x\in\mathbb Z_q^n:d_T(\bm x,\bm z)^2\le\rho\}.
		\]
		Counting pairs $(\bm z,\bm c)$ with
		$\bm z\in\mathbb Z_q^n$, $\bm c\in C$, and
		$\bm c\in B(\bm z,\rho)$ gives
		\[
		\frac1{q^n}\sum_{\bm z\in\mathbb Z_q^n}
		|C\cap B(\bm z,\rho)|
		=
		\frac{M V_q(n,\rho)}{q^n}.
		\]
		Since $|C\cap B(\bm z,\rho)|$ is an integer for each $\bm z$, there
		exists some $\bm z_0\in\mathbb Z_q^n$ such that
		\[
		|C\cap B(\bm z_0,\rho)|
		\ge
		\left\lceil
		\frac{M V_q(n,\rho)}{q^n}
		\right\rceil
		=
		L_\rho.
		\]
		Let $$S=C\cap B(\bm z_0,\rho)$$ and $|S|=N$, which implies
		$N\ge L_\rho$.	
			Since $S\subseteq B(\bm z_0,\rho)$, we have
		$d_T(\bm x,\bm z_0)^2\le \rho$ for every $\bm x\in S$, yielding
	\begin{equation}
		\label{dtlemudaq0}
		\sum_{\bm x\in S} d_T(\bm x,\bm z_0)^2\le N\rho.
	\end{equation}
		By Lemma~\ref{lem:lambda-toroidal-comparison}, we have
	\begin{equation}
			\label{dtlemudaqp}
			\sum_{\bm x,\bm y\in S} d_T(\bm x,\bm y)^2
		\le
		2 N\sum_{\bm x\in S} d_T(\bm x,\bm z_0)^2.
	\end{equation}
Combining \eqref{dtlemudaq0} and \eqref{dtlemudaqp}, we obtain
		\begin{equation}
			\label{ledt}
			\sum_{\bm x,\bm y\in S}d_T(\bm x,\bm y)^2
			\le
			2 N^2\rho.
		\end{equation}

		On the other hand, since every ordered pair
		$(\bm x,\bm y)\in S^2$ with $\bm x\ne\bm y$ contributes at
		least $d_T(C)^2$, we also have
	\begin{equation}
		\label{gedt}
		N(N-1)d_T(C)^2
		\le
		\sum_{\bm x,\bm y\in S}d_T(\bm x,\bm y)^2
		.
	\end{equation}
		Therefore we obtain $$N(N-1)d_T(C)^2\le2N^2\rho$$ by \eqref{ledt} and \eqref{gedt}, which implies 
\[
d_T(C)^2
\le
2\rho\frac{N}{N-1}
\le
2\rho\frac{L_\rho}{L_\rho-1}
\]
		since $N\ge L_\rho$. Taking the
		minimum over all $\rho$ with $L_\rho\ge2$ shows the result.
	\end{proof}

	\begin{example}
		Let
		\[
		C=\{(0,0),(2,0),(0,2),(2,2)\}\subseteq\mathbb Z_4^2.
		\]
		Then $|C|=4$ and a direct computation shows that $d_T(C)=2$.
		
		We show that this code is optimal by Theorem~\ref{thm:local-ball-plotkin}.
		For $q=4$, $n=2$, $M=4$, and $\rho=1$, the volume of $B(\bm 0,1)$ is
		\[
		V_4(2,1)=5.
		\]
		Thus
		\[
		L_\rho
		=
		\left\lceil\frac{M V_4(2,1)}{4^2}\right\rceil
		=
		\left\lceil\frac{4\cdot 5}{16}\right\rceil
		=
		2.
		\]
		By Theorem~\ref{thm:local-ball-plotkin}, for every code
		$C'\subseteq\mathbb Z_4^2$ with $|C'|=4$, we have
		\[
		d_T(C')^2
		\le
		2\rho\frac{L_\rho}{L_\rho-1}
		=
		4.
		\]
		Thus every $4$-point code $C'\subseteq\mathbb Z_4^2$ satisfies
		$d_T(C')\le2$. Since the above code $C$ satisfies $d_T(C)=2$, it
		attains the upper bound in Theorem~\ref{thm:local-ball-plotkin} and is
		therefore an MTD code.
	\end{example}
	\subsection{Delsarte Linear Programming Bound}	
	
The Delsarte linear programming method is based on the
association-scheme framework introduced by Delsarte~\cite{Delsarte73}.
For the Lee metric, Sol\'e studied the Lee association
scheme~\cite{Sole88}, and Astola and Tabus applied this method to linear Lee codes in~\cite{AstolaTabus}. In this section, we adapt this framework to the
toroidal metric by expressing the squared toroidal distance in terms of
Lee compositions.

	For \(\bm z=(z_1,\dots,z_n)\in\mathbb Z_q^n\), its \emph{Lee composition} is
	\[
	\operatorname{comp}_L(\bm z)
	=
	(\alpha_0,\alpha_1,\ldots,\alpha_{\lfloor q/2\rfloor}),
	\]
	where
	\[
	\alpha_r
	=
	\#\{1\le t\le n:d_L(z_t,0)=r\},
	\qquad 0\le r\le \lfloor q/2\rfloor.
	\]
	Let \(\mathcal L_q(n)\) denote the set of all Lee compositions in
	\(\mathbb Z_q^n\), namely
	\[
	\mathcal L_q(n)
	=
	\left\{
	(\alpha_0,\alpha_1,\ldots,\alpha_{\lfloor q/2\rfloor})
	\middle|
	\sum_{r=0}^{\lfloor q/2\rfloor}\alpha_r=n,\ 
	\alpha_r\in\mathbb Z_{\ge0}
	\right\}.
	\]
For \(\alpha,\beta\in\mathcal L_q(n)\), write
\[
\alpha=(\alpha_0,\alpha_1,\ldots,\alpha_{\lfloor q/2\rfloor}),
\qquad
\beta=(\beta_0,\beta_1,\ldots,\beta_{\lfloor q/2\rfloor}).
\]
Put \(\zeta=e^{2\pi \mathrm i/q}\). For \(0\le r\le \lfloor q/2\rfloor\), define
\[
P_r(y_0,y_1,\ldots,y_{\lfloor q/2\rfloor})
=
\sum_{a\in\mathbb Z_q} y_{d_L(a,0)}\zeta^{ar}.
\]
The \emph{Lee--Krawtchouk coefficient} \(K_\beta(\alpha)\) is 
\[
K_\beta(\alpha)
=
\left[
y_0^{\beta_0}y_1^{\beta_1}\cdots
y_{\lfloor q/2\rfloor}^{\beta_{\lfloor q/2\rfloor}}
\right]
\prod_{r=0}^{\lfloor q/2\rfloor}
P_r(y_0,y_1,\ldots,y_{\lfloor q/2\rfloor})^{\alpha_r},
\]
where
\[
\left[
y_0^{\beta_0}y_1^{\beta_1}\cdots
y_{\lfloor q/2\rfloor}^{\beta_{\lfloor q/2\rfloor}}
\right]f
\]
denotes the coefficient of
\(y_0^{\beta_0}y_1^{\beta_1}\cdots
y_{\lfloor q/2\rfloor}^{\beta_{\lfloor q/2\rfloor}}\) in the polynomial \(f\).

	For a code \(C\subseteq\mathbb Z_q^n\) with \(|C|=M\), and for each
	Lee composition \(\alpha\in\mathcal L_q(n)\), let
\[
A_\alpha(C)
=
\frac1{|C|}
\#\{(\bm x,\bm y)\in C^2 :
\operatorname{comp}_L(\bm x-\bm y)=\alpha\}.
\]
Since the ordered pairs \((\bm x,\bm y)\in C^2\) are partitioned according to the Lee composition of \(\bm x-\bm y\), we have
\[
\sum_{\alpha\in\mathcal{L}_q(n)}A_\alpha(C) = \frac{1}{M}|C|^2 = M.
\]
	Put	
	\[
	D(\alpha)
	=
	\sum_{r=1}^{\lfloor q/2\rfloor} r^2\alpha_r.
	\]
Thus, if \(\operatorname{comp}_L(\bm x-\bm y)=\alpha\), then
\[
d_T(\bm x,\bm y)^2=D(\alpha).
\]

	Let
$v_\alpha$ denote the number of vectors in $\mathbb Z_q^n$ with Lee composition
$\alpha$. Explicitly,
\[
v_\alpha
=
\binom{n}{\alpha_0,\alpha_1,\ldots,\alpha_{\lfloor \frac{q}{2}\rfloor}}
\prod_{r=0}^{\lfloor \frac{q}{2}\rfloor} m_r^{\alpha_r},
\]
where the multinomial coefficient is
\[
\binom{n}{\alpha_0,\alpha_1,\ldots,\alpha_{\lfloor q/2\rfloor}}
=
\frac{n!}{\alpha_0!\alpha_1!\cdots \alpha_{\lfloor q/2\rfloor}!}
\]
and
\[
m_r=\#\{a\in\mathbb Z_q:d_L(a,0)=r\}=
\begin{cases}
	1, & r=0 \text{ or } \left(r=\dfrac q2 \text{ and } q \text{ is even}\right),\\
	2, & 1\le r<\dfrac q2.
\end{cases}
\] 

For \(\delta\in\mathbb Z_{>0}\), let \(L_q(n,\delta)\) denote the optimal value of
\begin{equation}
\label{Lqndeta}
\begin{aligned}
	L_q(n,\delta)=\max\quad
	&\sum_{\alpha\in\mathcal L_q(n)}B_\alpha\\
	\text{subject to}\quad
	&\left\{
	\begin{aligned}
		&B_{(n,0,\ldots,0)}=1,\qquad 0\le B_\alpha\le v_\alpha,\\
		&B_\alpha=0\quad\text{\rm if }0<D(\alpha)<\delta,\\
		&\sum_{\alpha\in\mathcal L_q(n)}
		K_\beta(\alpha)B_\alpha\ge0
		\quad\text{\rm for all }\beta\in\mathcal L_q(n)
	\end{aligned}
	\right.
\end{aligned}
\end{equation}	
where all \(B_\alpha\) are real variables indexed by
\(\alpha\in\mathcal L_q(n)\) and \(K_\beta(\alpha)\) denotes the
Lee--Krawtchouk coefficient.

	\begin{theorem}[Delsarte LP Bound for Toroidal Distance]
		\label{thm:delsarte-lp-toroidal}
		Let \(C\subseteq \mathbb Z_q^n\) be a code with \(|C|=M\ge2\). Then
		\[
		d_T(C)\le \sqrt{\delta_0-1},
		\]
		where
		\[
		\delta_0=
		\min\{\delta\in\mathbb Z_{>0}: M>L_q(n,\delta)\},
		\]
		and \(L_q(n,\delta)\) is defined in~\eqref{Lqndeta}.
	\end{theorem}

	\begin{proof}
Since the Lee composition of $\bm 0$ is \((n,0,\ldots,0)\), we have
\[
A_{(n,0,\ldots,0)}(C)=1.
\]
 Also $A_\alpha(C)\ge0$ for every Lee composition $\alpha$.
 
		For each fixed
		$\bm x\in C$, there are at most $v_\alpha$ possible differences
		$\bm x-\bm y$ with Lee composition $\alpha$. Hence
\[
\begin{aligned}
	M A_\alpha(C)
	&=
	\left|
	\left\{
	(\bm x,\bm y)\in C^2:
	\operatorname{comp}_L(\bm x-\bm y)=\alpha
	\right\}
	\right| \\
	&\le
	\left|
	\left\{
	(\bm x,\bm y)\in C\times \mathbb Z_q^n:
	\operatorname{comp}_L(\bm x-\bm y)=\alpha
	\right\}
	\right| \\
	&=
	M v_\alpha
\end{aligned}
\]
		and  $A_\alpha(C)\le v_\alpha$.
				
Assume now that \(d_T(C)^2\ge\delta\). 
	There do not exist two distinct codewords $\bm x, \bm y \in C$  such that \(\operatorname{comp}_L(\bm x-\bm y)=\alpha\)  if	$0<D(\alpha)<\delta
	$, which implies
		\[
		A_\alpha(C)=0
		\qquad
		\text{whenever }0<D(\alpha)<\delta.
		\]
		
		It remains to justify the Delsarte positivity constraints. 
		For  \(\bm u, \bm z\in\mathbb Z_q^n\), define the additive character
		$$
		\chi_{\bm u}(\bm z)=\zeta^{\bm u\cdot \bm z},$$  where
			$\zeta=e^{2\pi \mathrm i/q}$.
		Then for every \(\bm u\in\mathbb Z_q^n\),
		\[
		\frac1M\sum_{\bm x,\bm y\in C}
		\chi_{\bm u}(\bm x-\bm y)
		=
		\frac1M\sum_{\bm x,\bm y\in C}
		\chi_{\bm u}(\bm x)\overline{\chi_{\bm u}(\bm y)}
		=
		\frac1M
		\left|
		\sum_{\bm x\in C}\chi_{\bm u}(\bm x)
		\right|^2
		\ge0.
		\]
	 Now fix a Lee composition \(\beta\in\mathcal L_q(n)\). Summing the preceding inequality over all characters \(\chi_{\bm u}\) with
		\(\operatorname{comp}_L(\bm u)=\beta\), we obtain
		\[
		\frac1M
		\sum_{\substack{\bm u\in\mathbb Z_q^n\\
				\operatorname{comp}_L(\bm u)=\beta}}
		\sum_{\bm x,\bm y\in C}
		\chi_{\bm u}(\bm x-\bm y)
		\ge0.
		\]
		Grouping the pairs \((\bm x,\bm y)\) according to
		\(\operatorname{comp}_L(\bm x-\bm y)=\alpha\), yields
		\[
		\sum_{\alpha\in\mathcal L_q(n)}
		\left(
		\sum_{\substack{\bm u\in\mathbb Z_q^n\\
				\operatorname{comp}_L(\bm u)=\beta}}
		\chi_{\bm u}(\bm z_\alpha)
		\right)
		A_\alpha(C)
		\ge0,
		\]
		where \(\bm z_\alpha\) is any vector of Lee composition \(\alpha\).  The
		inner sum is actually the Lee--Krawtchouk
		coefficient \(K_\beta(\alpha)\).  Hence
		\[
		\sum_{\alpha\in\mathcal L_q(n)}
		K_\beta(\alpha)A_\alpha(C)
		\ge0
		\qquad
		\text{for all }\beta\in\mathcal L_q(n).
		\]
		They are exactly the Delsarte positivity constraints for the Lee
		association scheme.

	We now relate the actual inner distribution of the code \(C\) to the
	linear program defining \(L_q(n,\delta)\).  The quantities
	\(B_\alpha\) in the linear program are abstract optimization variables,
	whereas \(A_\alpha(C)\) is the Lee inner distribution determined by the
	given code \(C\).  If \(d_T(C)^2\ge \delta\), then the preceding arguments
	show that the family \(\{A_\alpha(C)\}_{\alpha\in\mathcal L_q(n)}\)
	satisfies all constraints imposed on the variables \(B_\alpha\) in the
	linear program.  Indeed, we may take
	\[
	B_\alpha=A_\alpha(C)
	\qquad
	\text{for all }\alpha\in\mathcal L_q(n).
	\]
	With this choice, the objective value of the linear program is
	\[
	\sum_{\alpha\in\mathcal L_q(n)}B_\alpha
	=
	\sum_{\alpha\in\mathcal L_q(n)}A_\alpha(C)
	=
	M.
	\]
	Therefore, whenever \(d_T(C)^2\ge \delta\), the linear program has a
	feasible solution with objective value \(M\).  Hence
	\[
	M\le L_q(n,\delta).
	\]
	Consequently, if \(M>L_q(n,\delta)\), then no code of size \(M\) can
	satisfy \(d_T(C)^2\ge\delta\).  Since \(d_T(C)^2\) is an integer, it
	follows that
	\[
	d_T(C)^2\le \delta-1.
	\]
	Taking the smallest integer \(\delta\) for which \(M>L_q(n,\delta)\), we
	obtain
	\[
	d_T(C)^2\le \delta_0-1,
	\qquad
	\delta_0=
	\min\{\delta\in\mathbb Z_{>0}: M>L_q(n,\delta)\}.
	\]
	Thus
	\[
	d_T(C)\le \sqrt{\delta_0-1}.
	\]
	This completes the proof.

	\end{proof}

	\begin{example}
		Let \(n=8\), \(q=4\), and
		\[
		C=2E_8\cap \mathbb Z_4^8,
		\]
		where 
		\begin{equation}
			\label{E8}
		E_8
		=
		\left\{
		\bm x=(x_1,\dots,x_8)\in\mathbb R^8 :
		\begin{array}{l}
			\text{either }x_i\in\mathbb Z\ \text{for all }i,\\[1ex]
			\text{or }x_i\in\mathbb Z+\dfrac12\ \text{for all }i,
		\end{array}
		\ \text{and }\sum_{i=1}^8 x_i\equiv0\pmod 2
		\right\}.
		\end{equation} S.~Liu and A.~Sakzad proved in
		\cite{LiuSakzad2026} that this code satisfies
		\[
		M=|C|=2^8,\qquad d_T(C)=2\sqrt{2}.
		\]
		
		We now substitute these parameters into the Delsarte linear programming bound.
		For \(q=4\), each difference vector in \(\mathbb Z_4^8\) is classified by a
		Lee composition indexed by elements $\alpha$ in
		\[
	\mathcal L_4(8)
		:=\{(\alpha_0,\alpha_1,\alpha_2)\in \mathbb Z_{\ge 0}^3:
		\alpha_0+\alpha_1+\alpha_2= 8\},
		\]
		where \(\alpha_r\) counts the coordinates of Lee weight \(r\). The corresponding
		squared toroidal distance is
		\[
		D(\alpha)=\alpha_1+4\alpha_2,
		\]
		and 
		\[
		v_\alpha
		=
		2^{\alpha_1}\binom{8}{\alpha_1}\binom{8-\alpha_1}{\alpha_2}.
		\]
		In particular, \(A_{(8,0,0)}(C)=1\), and
		\[
		\sum_{\alpha\in\mathcal L_4(8)} A_\alpha(C)=M.
		\]
		
	For \(\delta\in\mathbb Z_{>0}\), the relevant Delsarte linear programming bound is
		\[
		\begin{aligned}
			L_4(8,\delta)
			=
			\max\quad
			&\sum_{\alpha\in\mathcal L_4(8)} B_\alpha \\
			\text{\rm subject to}\quad
			&\left\{
			\begin{aligned}
				&B_{(8,0,0)}=1,\qquad 0\le B_\alpha\le v_\alpha,\\
				&B_\alpha=0
				\qquad
				\text{\rm if }0<D(\alpha)<\delta,\\
				&\sum_{\alpha\in\mathcal L_4(8)}
				K_\beta(\alpha)B_\alpha\ge 0
				\qquad
				\text{\rm for all }\beta\in\mathcal L_4(8).
			\end{aligned}
			\right.
		\end{aligned}
		\]
		Here, for \(\alpha=(\alpha_0, \alpha_1,\alpha_2)\) and
		\(\beta=(\beta_0, \beta_1,\beta_2)\), the Krawtchouk coefficients are
		\[
		K_\beta(\alpha)
		=
		[X^{\beta_1}Y^{\beta_2}]
		(1+2X+Y)^{8-\alpha_1-\alpha_2}
		(1-Y)^{\alpha_1}
		(1-2X+Y)^{\alpha_2}.
		\]
		
	The optimal values of the linear programming bound for different values of
	\(\delta\) are shown below. 
		\begin{table}[htbp]
			\centering
			
			\renewcommand{\arraystretch}{1.5}
			\setlength{\tabcolsep}{18pt}
			
			\begin{tabular}{c|c}
				\hline
				$\delta$ & $L_4(8,\delta)$\\
				\hline
				7 & 320.8406574513\\
				8 & 256.0000000000\\
				9 & 106.1293559098\\
				\hline
			\end{tabular}
			
			\caption{Delsarte linear programming bounds for \(q=4\) and \(n=8\)}
			\label{tab:lp-bound}
		\end{table}

From Table~\ref{tab:lp-bound}, we obtain that 
$
\delta_0=
\min\{\delta\ge1 : 256>L_q(8,\delta)\}=9.
$		
		Therefore, for any code \(C'\subseteq\mathbb Z_4^8\) with $|C'|= 2^8$, the Delsarte linear programming bound gives
		\[
		d_T(C')\le \sqrt{\delta_0-1}=2\sqrt{2}.
		\]
		Since the code \(C=2E_8\cap\mathbb Z_4^8\) has
		\[
		M=256,\qquad d_T(C)^2=8,
		\]
		it attains the Delsarte linear programming bound.
	\end{example}

\section{Maximum Toroidal Distance (MTD) codes}
\label{MTDcode}

\begin{definition}[Maximum Toroidal Distance (MTD) codes]
	Let \(q\ge 2\) and \(\ell\ge 1\) be integers, where \(\ell\) is a power of \(2\).
	Let
	\[
	C\subseteq \mathbb Z_q^\ell,\qquad |C|=2^\ell .
	\]
	Then \(C\) is called a \emph{maximum toroidal distance code}
	(\emph{MTD code}) if
	\[
	C
	\in
	\operatorname*{arg\,max}_{\substack{ C'\subseteq \mathbb Z_q^\ell\\
			|C'|=2^\ell}}
	d_T(C').
	\]
\end{definition}
Equivalently, an MTD code is a codebook of size \(2^\ell\) in
\(\mathbb Z_q^\ell\) whose minimum toroidal distance is as large as possible.

 MTD codes for small parameters $\ell$ and $q$ can be found using exhaustive search algorithms, which require approximately \(\binom{q^\ell}{2^\ell}\) operations. 
	As \(q\) and \(\ell\) increase, the computational cost of exhaustive search algorithms grows exponentially. In Table~\ref{tabl1}, we list the minimum toroidal distances of several MTD codes obtained by exhaustive search.
	
	\begin{table}[htbp]
	
		\centering
		\normalsize
		\renewcommand{\arraystretch}{1.0}
			
		\begin{tabular}{c c c c c c c c}
						\hline
			$d_T$ & $q=2$ & $q=3$ & $q=4$ & $q=5$ & $q=6$ & $q=7$& $q=8$ \\
			\hline
			$\ell=1$ & 1 & 1 & 2 & 2 & 3 & 3 & 4 \\
			\hline
			$\ell=2$ & 1  & 1  & 2  & $\sqrt{5}$ &  3 & $\sqrt{10}$ & $\sqrt{17}$  \\
			\hline
			$\ell=4$ & 1  &  $\sqrt{2}$ & 2  & $\sqrt{6}$  & 2$\sqrt{3}$  &  &   \\
			\hline
						
		\end{tabular}
		\caption{Minimum toroidal distances $d_T$ of MTD codes for different $\ell$ and $q$}
		\label{tabl1}
	\end{table}

	For the cases \(\ell=4\), \(q=7\) and \(q=8\), the exhaustive search
	algorithm fails to return results in a reasonable time. Therefore, we try to find other methods for constructing MTD codes.

	\subsection{$\ell=2$}
	\begin{theorem}[MTD codes for \(\ell=2\)]
		Let \(q\ge 2\), and write
		\(
		a=\left\lfloor\frac q2\right\rfloor
.
		\)
		Among all \(4\)-point codes \(\mathcal C\subseteq\mathbb Z_q^2\), there exists an MTD code which, up to translation and torus isometry, is of the form
		\[
		\mathcal C
		=
		\{(0,0),(\gamma^*,a),(a,q-\gamma^*),(a+\gamma^*,a-\gamma^*)\},
		\]
		where
		\[
		\gamma^*
		\in
		\operatorname*{arg\,max}_{\gamma\in
			\{\lfloor   q-\sqrt{q^2-(q-a)^2} \rfloor,\lfloor q-\sqrt{q^2-(q-a)^2} \rfloor+1\}}
		\min\left\{
		\sqrt{a^2+\gamma^2},
		\sqrt{(a+\gamma-q)^2+(\gamma-a)^2}
		\right\}.
		\]
	\end{theorem}

\begin{proof}
	
	Assume that \(C=\{\bm 0,\bm u,\bm v,\bm w\}\) is an MTD codebook. By translation invariance, we may assume that the pair \(\bm 0\) and \(\bm u\) attains the minimum toroidal distance of the code.
Let
\[
d_T(C)=d_T(\bm 0,\bm u)=|\bm u|_T.
\]
Then every other pair of codewords must have toroidal distance at least \(|\bm u|_T\). We now describe the corresponding feasible region geometrically.

  \begin{enumerate}[label=(\arabic*)]
  	\item We first consider the case $q$ even.
  		Identify \(\mathbb Z_q^2\) with the set of integer points in the square \([0,q)\times[0,q)\), and divide this square into four subsquares of side length \(\frac{q}{2}\). Since we want \(|\bm u|_T\) to be as large as possible while keeping the feasible region as large as possible, we write
  	\[
  	\bm u=(\gamma,a), \qquad a=\frac q2,\qquad 0\le \gamma<a,
  	\]
  	where the parameter \(\gamma\) is chosen to maximize \(|\bm u|_T\).
  	
  	In each subsquare, the toroidal distance from a point to \(\bm 0\) is the Euclidean distance to the corresponding corner among
  	\[
  	(0,0),\ (0,q),\ (q,0),\ (q,q).
  	\]
  	Now draw the circle centered at \(\bm u\) with radius \(|\bm u|_T\). Any other codeword must lie outside this circle. Also draw the circles of the same radius centered at
  	\[
  	(0,0),\ (0,q),\ (q,0),\ (q,q),\\ 
  	\bm u+(0,q) , \bm u+(q,0), \bm u+(q,q).
  	\]
  	Hence every admissible position of another codeword must lie inside the square \([0,q)\times[0,q)\) and outside all these circles. This determines the feasible region. In order to place two further codewords in this region, its toroidal diameter must be at least \(|\bm u|_T\).
  	
  	We begin with \(\gamma=0\), and then vary \(\gamma\) in an attempt to enlarge \(|\bm u|_T\). Suppose that for some \(\gamma^*\), the toroidal diameter of the feasible region is exactly \(|\bm u|_T\). We then take \(\gamma=\gamma^*\). A direct computation shows that when \(\gamma>0\), the feasible region splits into an upper part and a lower part. The upper part has extreme points
  	\[
  	\left(\frac q2,q-\gamma\right), \qquad \left(\frac q2+\gamma,\frac q2+\gamma\right),
  	\]
  	while the lower part has extreme points
  	\[
  	\left(\frac q2+\gamma,\frac q2-\gamma\right), \qquad \left(\frac q2,\gamma\right).
  	\]
  	Moreover,
  	\[
  	d_T\!\left(\left(\frac q2,q-\gamma\right),\left(\frac q2+\gamma,\frac q2-\gamma\right)\right)=|\bm u|_T,
  	\]
  	and
  	\[
  	d_T\!\left(\left(\frac q2+\gamma,\frac q2+\gamma\right),\left(\frac q2,\gamma\right)\right)=|\bm u|_T.
  	\]
  	Thus each of these pairs realizes the diameter of the feasible region. Therefore, we may choose \(\bm v\) and \(\bm w\) to be one of these pairs.
  	
  	This leads to the following two candidate codebooks:
  	\[
  	\bm{u} = (\gamma,a), \quad \bm{v} = (a,q-\gamma), \quad 
  	\mathcal{C} = \{(0,0),(\gamma,a),(a,q-\gamma),(a+\gamma,a-\gamma)\},
  	\]
  	or
  	\[
  	\bm{u}' = (\gamma,a), \quad \bm{v}' = (a,\gamma), \quad 
  	\mathcal{C}' = \{(0,0),(\gamma,a),(a,\gamma),(a+\gamma,a+\gamma)\}.
  	\]
  	
  	It remains to compare their minimum toroidal distances. We have
  	\[
  	d_T(\mathcal C)
  	=\min\!\left\{\sqrt{a^2+\gamma^2},\ \sqrt{(\gamma-a)^2+(a+\gamma-q)^2}\right\},
  	\]
  	and
  	\[
  	d_T(\mathcal C')
  	=\min\!\left\{\sqrt{a^2+\gamma^2},\ \sqrt{(\gamma-a)^2+(a-\gamma)^2}\right\}.
  	\]
  	Therefore it suffices to compare \((a+\gamma-q)^2\) with \((a-\gamma)^2\). Since
  	\[
  	(a+\gamma-q)^2-(a-\gamma)^2=(2a-q)(2\gamma-q)\ge 0,
  	\]
  	we obtain
  	\[
  	d_T(\mathcal C)\ge d_T(\mathcal C').
  	\]
  	Hence a better choice is
  	\[
  	\mathcal{C}=\{(0,0),(\gamma,a),(a,q-\gamma),(a+\gamma,a-\gamma)\}.
  	\]

  	We know $d_T(\mathcal{C})=\min(\sqrt{a^2+\gamma^2},\sqrt{(a+\gamma-q)^2+(\gamma-a)^2})$. In the case $\sqrt{a^2+\gamma^2}=\sqrt{(a+\gamma-q)^2+(\gamma-a)^2}$, we obtain $\gamma=q-\sqrt{q^2-(a-q)^2}=(1-\frac{\sqrt{3}}{2})q$. Since $\gamma$ is an integer, we have 
  	\[
  	\gamma^* \in
  	\operatorname*{arg\,max}_{\gamma\in
  		\left\{
  		\left\lfloor\left(1-\frac{\sqrt{3}}{2}\right)q\right\rfloor,\,
  		\left\lfloor\left(1-\frac{\sqrt{3}}{2}\right)q\right\rfloor+1
  		\right\}}
  	\min\left\{
  	\sqrt{a^2+\gamma^2},
  	\sqrt{(a+\gamma-q)^2+(\gamma-a)^2}
  	\right\}.
  	\]
  	
  	\item 	The proof of case $q$ odd is analogous to the even case, so we only indicate the difference. By translation invariance, we may assume that
  	\[
  	d_T(\mathcal C)=d_T(0,\bm u)=|\bm u|_T,
  	\qquad
  	\bm u=(\gamma,a),\quad a=\frac{q-1}{2},\quad 0\le \gamma\le a.
  	\]
  	As in the even case, every other codeword must lie in the feasible region determined by the condition that all the toroidal distances from \(0\), \(\bm u\), and the relevant torus translates of these points are at least \(|\bm u|_T\).
  	
  	Compared with the even case, the only geometric difference is that the
  	middle lines \(x=q/2\) and \(y=q/2\) are replaced by two adjacent integer
  	lines in each coordinate direction:
  	\[
  	x=\frac{q-1}{2},\quad x=\frac{q+1}{2};
  	\qquad
  	y=\frac{q-1}{2},\quad y=\frac{q+1}{2},
  	\]
  	which produce one vertical strip and one
  	horizontal strip of width \(1\); see Figure~\ref{fig:oddfeasible-region}.  Since these strips contain no integer
  	points, they cannot be part of the feasible region.
  	However, this does not change the extremal configuration: the diameter is still attained by boundary points in the same relative position as in the even case. Hence the remaining two codewords may again be chosen as
  	\[
  	(\frac{q-1}{2},q-\gamma)
  	\qquad\text{and}\qquad
  	(\frac{q-1}{2}+\gamma,\frac{q-1}{2}-\gamma).
  	\]
  	Therefore an MTD code may be taken in the stated form. 
  	
  	When $q=3$, it is worth noting that the feasible region can also be attained on the $y$-axis, as shown in Figure~\ref{fig:q3feasible-region}. Hence $\bm v$ and $\bm w$ may also be taken as $(0,2)$ and $(1,2)$.

  	The optimal $\gamma$ can also be found in a similar way.
  	$$	\gamma^*\in \operatorname*{arg\,max}_{\gamma\in
  		\left\{
  		\left\lfloor q-\sqrt{q^2-(a-q)^2}\right\rfloor,\,
  		\left\lfloor q-\sqrt{q^2-(a-q)^2}\right\rfloor+1
  		\right\}} \min(\sqrt{a^2+\gamma^2},\sqrt{(a+\gamma-q)^2+(\gamma-a)^2}).$$
  \end{enumerate}

	\end{proof}

\begin{center}
	\begin{minipage}[t]{0.45\textwidth}
		\centering
		\includegraphics[width=\textwidth]{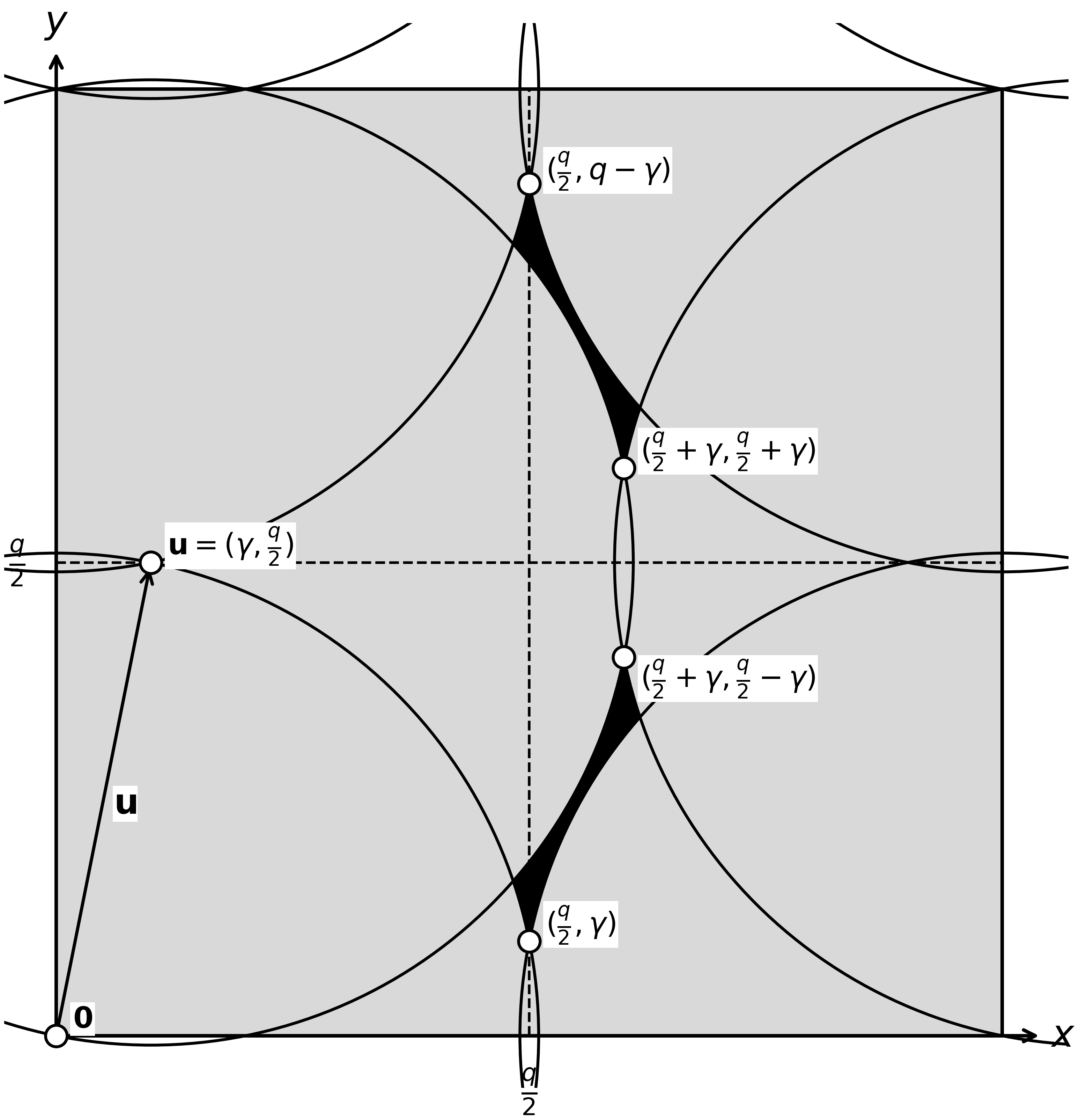}
		\captionof{figure}{Feasible region of codewords for even $q$.}
		\label{fig:feasible-region}
	\end{minipage}
	\hfill
	\begin{minipage}[t]{0.45\textwidth}
		\centering
		\includegraphics[width=\textwidth]{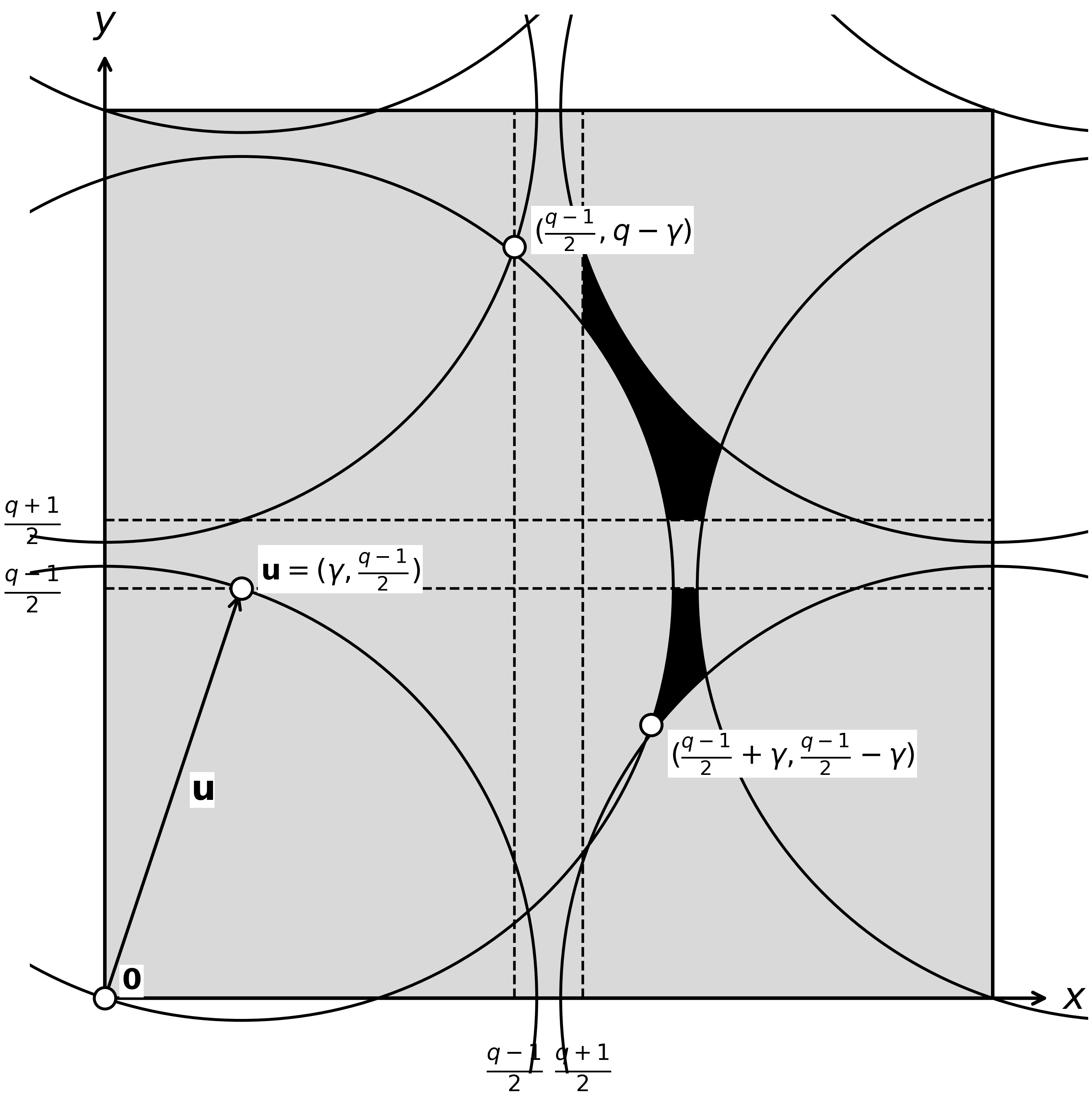}
		\captionof{figure}{Feasible region of codewords for odd $q$.}
		\label{fig:oddfeasible-region}
	\end{minipage}
\end{center}
	
\begin{figure}[!htb]
	\centering
	\includegraphics[width=0.6\textwidth]{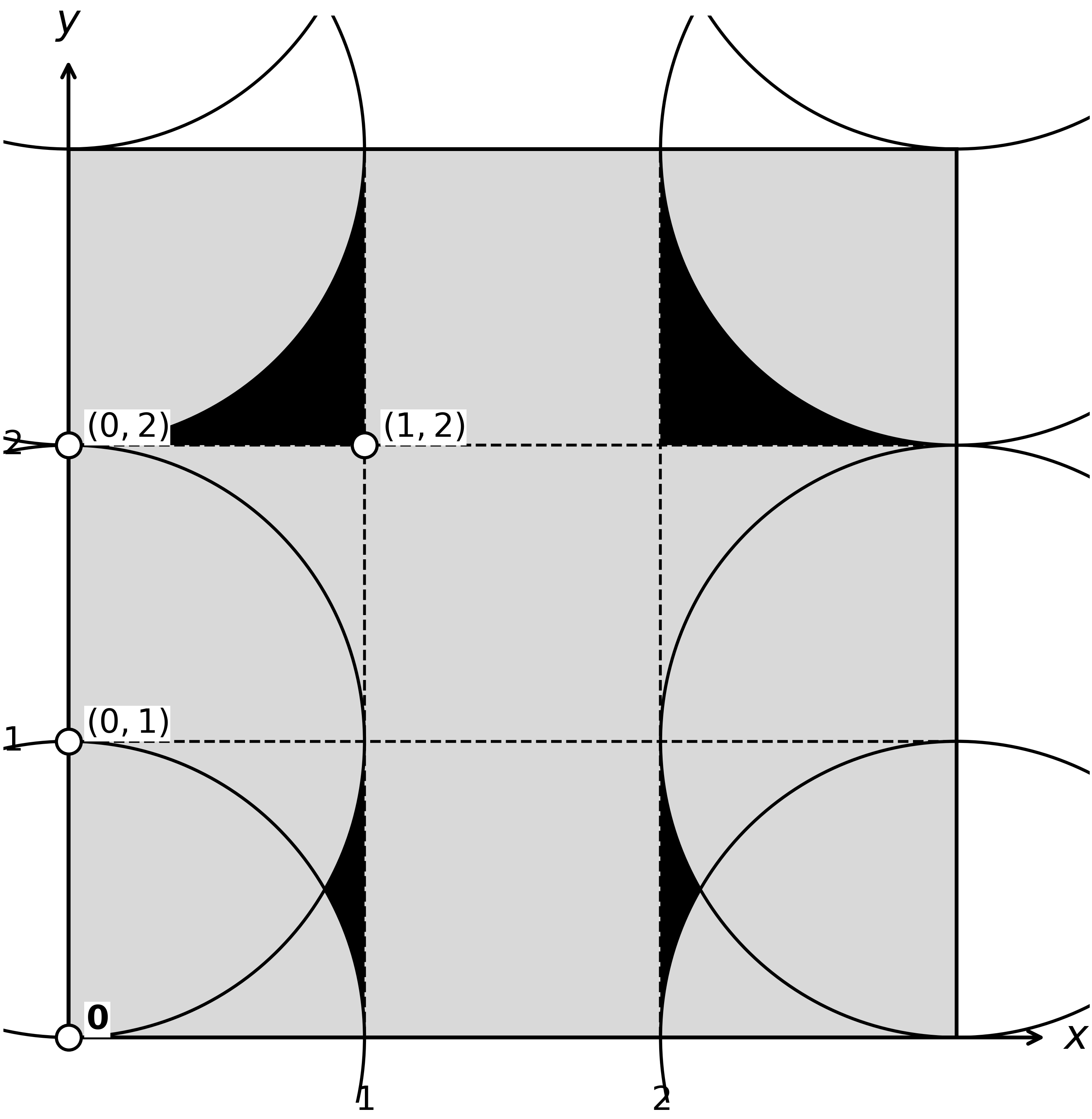}
	\caption{Feasible region of codewords for odd \(q=3\).}
	\label{fig:q3feasible-region}
\end{figure}

\subsection{$\ell=4$}

For the case $\ell=4$, we generalize the construction from a generator
matrix of MTD code for $\ell=2$. This yields a rather strong construction.
For several special $q$, we find that the resulting codes are all
MTD codes, which provides indirect evidence that this family is a strong
candidate for MTD constructions.

\begin{theorem}
		\label{l441}
Let \(q=2a\) be even, and let
\[
G
=
\begin{pmatrix}
	a & \gamma & \gamma & \gamma\\
	2a-\gamma & a & \gamma & 2a-\gamma\\
	2a-\gamma & 2a-\gamma & a & \gamma\\
	2a-\gamma & \gamma & 2a-\gamma & a
\end{pmatrix},
\]
where \(\gamma\in\{0,1,\ldots,a\}\). We consider the \(16\)-point code
\[
C
=
\{\bm m G\bmod q:\ \bm m\in\{0,1\}^4\}
\subseteq \mathbb Z_q^4.
\]
Set
	\[
\begin{aligned}
	f_1&=a^2+3\gamma^2;\\
	f_2&=
	\begin{cases}
		2(a^2-2a\gamma+3\gamma^2), & 0\le\gamma\le a/2,\\
		6(a-\gamma)^2, & a/2\le\gamma\le a;
	\end{cases}\\
	f_3&=3a^2-8a\gamma+9\gamma^2;\\
	f_4&=
	\begin{cases}
		3(a^2+3\gamma^2), & 0\le\gamma\le a/3,\\
		7a^2-12a\gamma+9\gamma^2, & a/3\le\gamma\le a;
	\end{cases}\\
	f_5&=4(a^2-3a\gamma+3\gamma^2).
\end{aligned}
\]
Then
\[
d_T(C)=\sqrt{\min\{f_1,f_2,f_3,f_4,f_5\}}.
\]
Consequently, within this family, the best choice of \(\gamma\) is any
\(\gamma^*\) satisfying
\[
\gamma^*
\in
\operatorname*{arg\,max}_{\substack{0\le \gamma\le a\\ \gamma\in\mathbb Z}}
\min\{f_1,f_2,f_3,f_4,f_5\}.
\]

\end{theorem}

\begin{proof}
	For distinct messages \(\bm m,\bm m'\in\{0,1\}^4\),  \( \bm m''=
	\bm m-\bm m'\) belongs to
	\[
	\Delta=\{-1,0,1\}^4\setminus\{\bm 0\}.
	\]
	For each \(\bm m''\in\Delta\), the squared toroidal distance
	between the corresponding pair of codewords is
	\[
	|\bm m'' G|_T^2.
	\]
	Although \(|\Delta|=3^4-1=80\), these differences fall into five distance
	types under the symmetries of the construction. Thus it suffices to compute
	\(|\bm m'' G|_T^2\) for the following representatives:
	\[
	(1,0,0,0),\quad
	(1,1,0,0),\quad
	(1,1,1,0),\quad
	(0,1,1,1),\quad
	(1,1,1,1).
	\]
	For these representatives, the corresponding squared toroidal distances are
	\[
	\begin{aligned}
		f_1&=a^2+3\gamma^2;\\
		f_2&=
		\begin{cases}
			2(a^2-2a\gamma+3\gamma^2), & 0\le\gamma\le a/2,\\
			6(a-\gamma)^2, & a/2\le\gamma\le a;
		\end{cases}\\
		f_3&=3a^2-8a\gamma+9\gamma^2;\\
		f_4&=
		\begin{cases}
			3(a^2+3\gamma^2), & 0\le\gamma\le a/3,\\
			7a^2-12a\gamma+9\gamma^2, & a/3\le\gamma\le a;
		\end{cases}\\
		f_5&=4(a^2-3a\gamma+3\gamma^2).
	\end{aligned}
	\]
	Hence
	\[
	d_T(C)^2=\min\{f_1,f_2,f_3,f_4,f_5\}.
	\]

\end{proof}

\begin{corollary}
	\label{l4444}
	Let \(q=2a\) with \(3\mid a\), and let \(C\) be the code in
	Theorem~\ref{l441} with \(\gamma=a/3\). Then
	\[
	d_T(C)=\frac{q}{\sqrt{3}}.
	\]
\end{corollary}

\begin{proof}
	By Theorem~\ref{l441}, the squared minimum toroidal distance is given by
	\[
	d_T(C)^2
	=
	\max_{\gamma\in\{0,1,\ldots,a\}}
	\min\{f_1(\gamma),f_2(\gamma),f_3(\gamma),f_4(\gamma),f_5(\gamma)\},
	\]
	where $f_i$, $i=1,\ldots,5$, are the functions defined in
	Theorem~\ref{l441}. We show that the maximum is attained at
	$\gamma=a/3$.
	
	Put
	\[
	H(\gamma)=\min\{f_1(\gamma),f_2(\gamma),f_3(\gamma),f_4(\gamma),f_5(\gamma)\}.
	\]
	For $0\le \gamma\le a/3$, we have
	\[
	H(\gamma)\le f_1(\gamma)=a^2+3\gamma^2\le \frac{4a^2}{3}.
	\]
	For $a/3\le \gamma\le 2a/3$,
	\[
	H(\gamma)\le f_5(\gamma)=4(a^2-3a\gamma+3\gamma^2).
	\]
	Moreover,
	\[
	f_5(\gamma)-\frac{4a^2}{3}
	=
	12\left(\gamma-\frac a3\right)
	\left(\gamma-\frac{2a}{3}\right)\le 0,
	\]
	and therefore $H(\gamma)\le 4a^2/3$ for $a/3\le \gamma\le 2a/3$. Finally, for
	$2a/3\le \gamma\le a$, since $\gamma\ge a/2$, the second branch of
	$f_2$ gives
	\[
	H(\gamma)\le f_2(\gamma)=6(a-\gamma)^2
	\le 6\left(\frac a3\right)^2
	=
	\frac{2a^2}{3}.
	\]
	Hence, for every admissible integer $\gamma$,
	\[
	H(\gamma)\le \frac{4a^2}{3}.
	\]
	
	Since $3\mid a$, $\gamma=a/3$ is an admissible integer. Therefore
	\[
	f_1=f_2=f_3=f_5=\frac{4a^2}{3},
	\qquad
	f_4=4a^2.
	\]
	Thus
	\[
	\max_{\gamma\in\{0,1,\ldots,a\}}H(\gamma)=\frac{4a^2}{3},
	\]
	which implies
	\[
	d_T(C)^2=\frac{4a^2}{3}.
	\]
	Hence
	\[
	d_T(C)=\frac{2a}{\sqrt{3}}=\frac{q}{\sqrt{3}}.
	\]
\end{proof}

We give examples for $q=4,6,12$ and prove that they are all MTD codes.

\begin{example}
	Let $q=4$. Then $a=2$. Taking $\gamma=1$, the matrix in
	Theorem~\ref{l441} becomes
	\[
	G=
	\begin{pmatrix}
		2&1&1&1\\
		3&2&1&3\\
		3&3&2&1\\
		3&1&3&2
	\end{pmatrix}.
	\]
	Define
	\[
	C=\{\bm mG \bmod 4:\bm m\in\{0,1\}^4\}\subseteq \mathbb Z_4^4.
	\]
	By Theorem~\ref{l441}, we have
	\[
	d_T(C)^2=\min\{f_1,f_2,f_3,f_4,f_5\}.
	\]
	Substituting $a=2$ and $\gamma=1$ gives
	\[
	f_1=7,\qquad
	f_2=6,\qquad
	f_3=5,\qquad
	f_4=13,\qquad
	f_5=4.
	\]
	Hence
	\[
	d_T(C)=2.
	\]
  Table \ref{tabl1} shows that every \(16\)-point
	code in \(\mathbb Z_4^4\) satisfies
	\[
	d_T\le 2.
	\]
	Since \(d_T(C)=2\), the code \(C\) attains this upper bound, which implies \(C\) is an MTD code.
\end{example}

\begin{example}
	For \(q=6\), we have \(a=3\) and \(\gamma=1\). Let
	\[
	G
	=
	\begin{pmatrix}
		3 & 1 & 1 & 1\\
		5 & 3 & 1 & 5\\
		5 & 5 & 3 & 1\\
		5 & 1 & 5 & 3
	\end{pmatrix},
	\]
	and define
	\[
	C:=\{\bm mG \bmod 6:\bm m\in\{0,1\}^4\}
	\subseteq\mathbb Z_6^4.
	\]
	By Corollary~\ref{l4444}, we have
	\[
	d_T(C)=2\sqrt{3}.
	\]
	Moreover, the Delsarte linear programming bound
	(Theorem~\ref{thm:delsarte-lp-toroidal}) shows that every \(16\)-point
	code in \(\mathbb Z_6^4\) satisfies
	\[
	d_T^2\le 12.
	\]
	Since \(d_T(C)^2=12\), the code \(C\) attains this upper bound, which implies that
	\(C\) is an MTD code.
\end{example}

\begin{example}
	For \(q=12\), we have \(a=6\) and \(\gamma=2\). Let
	\[
	G
	=
	\begin{pmatrix}
		6 & 2 & 2 & 2\\
		10 & 6 & 2 & 10\\
		10 & 10 & 6 & 2\\
		10 & 2 & 10 & 6
	\end{pmatrix},
	\]
	and define
	\[
	C:=\{\bm mG\bmod 12:\bm m\in\{0,1\}^4\}
	\subseteq\mathbb Z_{12}^4.
	\]
	By Corollary~\ref{l4444}, we have
	\[
	d_T(C)=4\sqrt{3}.
	\]
	Moreover, the Delsarte linear programming bound
(Theorem~\ref{thm:delsarte-lp-toroidal}) shows that every \(16\)-point
code in \(\mathbb Z_{12}^4\) satisfies
	\[
	d_T^2\le 48.
	\]
	Since \(d_T(C)^2=48\), the code \(C\) attains this upper bound.
	Therefore \(C\) is an MTD code.
\end{example}

	\subsection{$\ell=8$}

Different sources may use different conventions for the determinant of a lattice $\Lambda$. Here, the determinant of $\Lambda$ is defined as the covolume of
$\Lambda$.
\begin{definition}[Determinant of a lattice]
	\label{def:lattice-determinant}
	Let $\Lambda\subset\mathbb R^\ell$ be a full-rank lattice, and let
	$B=(\bm b_1,\ldots,\bm b_\ell)$ be a basis matrix of $\Lambda$, so that
	\[
	\Lambda=B\mathbb Z^\ell.
	\]
The determinant of $\Lambda$ is defined as the volume of a fundamental
region of $\Lambda$, namely
	\[
	\det(\Lambda):=\operatorname{vol}(\mathbb R^\ell/\Lambda)=|\det B|.
	\]
\end{definition}

\begin{lemma}[{\cite[Chap.~1, Sec.~1.2; Chap.~8, Sec.~3]{ConwaySloane1999}}]
	\label{lmlattie1}
	Let \(\Lambda\subset \mathbb R^\ell\) be a full-rank lattice such that
	\(
	p\mathbb Z^\ell\subseteq \Lambda\subseteq \mathbb Z^\ell .
	\)
	Then the number of points in \(\Lambda \cap \mathbb Z_p^\ell\) is
	\[
	N=\left|\Lambda \cap \mathbb Z_p^\ell\right|=\frac{p^\ell}{\det(\Lambda)},
	\]
	where \(\det(\Lambda)\) denotes the determinant of \(\Lambda\).
\end{lemma}

	\begin{lemma}[{\cite[Lemma~2]{LiuSakzad2026}}]
		\label{lmlattie2}
		Let $\Lambda \subseteq \mathbb Z^\ell$ be an integer lattice containing the sublattice $p\mathbb Z^\ell$. Then its minimum Euclidean distance $d_{\min}(\Lambda)$ satisfies
		\[
		d_{\min}(\Lambda)
		=\min_{\substack{\bm v_1,\bm v_2\in \Lambda\\ \bm v_1\ne \bm v_2}}
		\|\bm v_1-\bm v_2\|
		= d_T\!\left(\Lambda \cap \mathbb Z_p^\ell\right).
		\]
	\end{lemma}

	The following construction in the case $m=1$ is given in \cite{LiuSakzad2026}, but it is not shown to be MTD. Here we generalize it to any $m > 0$  and  show that they are MTD codes.

\begin{theorem}
	Let \(\ell=8\) and let $q=4m$ with $m\in\mathbb Z_{>0}$. Set
	\[
	\Lambda=2mE_8\subset\mathbb R^8,
	\qquad
	C=\Lambda\cap\mathbb Z_q^8,
	\]
where $E_8$ is defined in \eqref{E8}.	Then $C$ is an MTD code in $\mathbb Z_q^8$ with
	\[
	d_T(C)=2\sqrt{2}\,m=\frac{q}{\sqrt{2}}.
	\]
\end{theorem}

\begin{proof}
	Since $2E_8\supseteq 4\mathbb Z^8$, we have
	\[
	\Lambda=2mE_8\supseteq 4m\mathbb Z^8=q\mathbb Z^8.
	\]
	Moreover, using $\det(E_8)=1$, we obtain
	\[
	\det(\Lambda)=(2m)^8=\frac{q^8}{2^8}.
	\]
	Thus, by Lemma~\ref{lmlattie1},
	\[
	|C|=|\Lambda\cap\mathbb Z_q^8|
	=\frac{q^8}{\det(\Lambda)}
	=2^8.
	\]
	By Lemma~\ref{lmlattie2}, $d_T(C)=d_{\min}(\Lambda)$. Since
	$\Lambda=2mE_8$ and $d_{\min}(E_8)=\sqrt{2}$, it follows that
	\[
	d_T(C)=d_{\min}(\Lambda)
	=2m\sqrt{2}
	=\frac{q}{\sqrt{2}}.
	\]
	
Now we show that $C$ is optimal. Let
	$C'\subseteq\mathbb Z_q^8$ be any code with $|C'|=2^8$, and let
	\[
	\mathcal P(C')=C'+q\mathbb Z^8\subseteq\mathbb R^8
	\]
	be its periodic extension. Each fundamental cell of $q\mathbb Z^8$
	contains exactly $2^8$ points of $\mathcal P(C')$, so the per-point volume is
	\[
	\frac{q^8}{2^8}=\det(\Lambda).
	\]
	By Viazovska's sphere-packing optimality theorem for the \(E_8\) lattice
	\cite{Viazovska2017}, among all point sets in \(\mathbb R^8\) with per-point
	volume \(\det(\Lambda)\), the minimum Euclidean distance is at most that of
	the scaled \(E_8\) lattice \(\Lambda\). Hence
	\[
	d_{\min}(\mathcal P(C'))
	\le d_{\min}(\Lambda)
	=
	\frac{q}{\sqrt2}.
	\]
	On the other hand,
	\[
	d_{\min}(\mathcal P(C'))=\min\{q,d_T(C')\}.
	\]
	Since $q/\sqrt{2}<q$, we obtain
	\[
	d_T(C')\le \frac{q}{\sqrt{2}}=d_T(C).
	\]
	Therefore $C$ is an MTD code in $\mathbb Z_q^8$.
\end{proof}

\subsection{$\ell=16$}

\begin{example}
	Let \(G_0\) be the following \(8\times 16\) matrix over \(\mathbb Z_4\):
	\[
	G_0=
	\left(
	\begin{array}{cccccccc|cccccccc}
		1&0&0&0&0&0&0&0&2&1&1&1&1&1&1&1\\
		0&1&0&0&0&0&0&0&1&2&3&1&2&1&0&0\\
		0&0&1&0&0&0&0&0&1&0&2&3&1&2&1&0\\
		0&0&0&1&0&0&0&0&1&0&0&2&3&1&2&1\\
		0&0&0&0&1&0&0&0&1&1&0&0&2&3&1&2\\
		0&0&0&0&0&1&0&0&1&2&1&0&0&2&3&1\\
		0&0&0&0&0&0&1&0&1&1&2&1&0&0&2&3\\
		0&0&0&0&0&0&0&1&1&3&1&2&1&0&0&2
	\end{array}
	\right).
	\]
	Define a \(16\times 16\) matrix
	\[
	\widetilde G=
	\begin{pmatrix}
		G_0\\
		2G_0
	\end{pmatrix}
	\pmod 4.
	\]
	Consider the code
	\[
	C=\{\bm m\widetilde G\pmod 4:\bm m\in\{0,1\}^{16}\}
	\subseteq \mathbb Z_4^{16}.
	\]
Then
	\[
	d_T(C)=3.
	\]
\end{example}

In the case \(\ell=16\), the largest minimum toroidal distance that we have
currently found is only \(3\). On the other hand, Theorem~\ref{thm:delsarte-lp-toroidal}
gives a theoretical upper bound
\[
d_T(C)\le 2\sqrt{3}.
\]
Whether there exists a code attaining this Delsarte LP bound for \(\ell=16\)
remains an open problem.

\section{Conclusion}

In this paper, we studied the minimum toroidal distance of codes over $\mathbb Z_q^\ell$ and its maximization problem. We established several upper bounds, including a Plotkin-type bound, a local ball--Plotkin bound, and the Delsarte linear programming bound, providing a systematic framework for analyzing $d_T(C)$. 
On the constructive part, we presented explicit code families for $\ell=2,4,8$, and proved their optimality in several cases, showing that they are MTD codes. In particular, the constructions for $\ell=4$ and $\ell=8$ demonstrate that algebraic and geometric methods can effectively yield optimal or near-optimal solutions. 
For $\ell=16$, we constructed a code over $\mathbb Z_4$ whose 
$d_T$ is relatively large with $q=4$. However, there remains a gap between the best construction currently known and
the Delsarte LP upper bound.

\FloatBarrier

\end{document}